\documentclass[pmlr]{jmlr}
\usepackage{graphicx}
\usepackage{amsmath,amssymb}
\usepackage{hyperref}
\usepackage{comment}
\usepackage{cleveref}
\usepackage{tikz}
\usetikzlibrary{automata, positioning, arrows.meta, decorations.pathreplacing, calc}
\definecolor{lred}{RGB}{255, 96, 96}
\definecolor{lgreen}{RGB}{64, 255, 64}
\definecolor{lblue}{RGB}{160, 160, 255}
\definecolor{gb}{RGB}{0,128,64}
\definecolor{rb}{RGB}{128,0,64}

\newcommand{\mrm}[1]{\mathrm{#1}}
\newcommand{\mbb}[1]{\mathbb{#1}}

\newcommand{\tsc}[1]{\textsc{#1}}

\newcommand{\lrceil}[1]{\lceil #1 \rceil}
\newcommand{\lrfloor}[1]{\lfloor #1 \rfloor}
\newcommand{\pre}{\mrm{Pref}}

\newcommand{\Str}{\mathit{Str}}
\newcommand{\Con}[2][]{{\normalfont\textsc{Con{#2}}\textsuperscript{#1}}}

\theorembodyfont{\upshape}
\newtheorem{problem}{Problem}

\newtheorem{exmpl}{Example}
\theoremheaderfont{\scshape}
\theorempostheader{:}
\theoremsep{\newline}

\jmlrproceedings{}{}
\jmlrpages{}

\title[Learning deterministic FSMs from the prefixes of a single string is NP-complete]{Learning deterministic finite-state machines from the prefixes of a single string is NP-complete}

\author{
  \Name{Radu Cosmin Dumitru} \Email{r.c.dumitru@student.tudelft.nl}\\
  \addr{Delft University of Technology, Delft, Netherlands}\\
  \Name{Ryo Yoshinaka} \Email{ryoshinaka@tohoku.ac.jp}\\
  \Name{Ayumi Shinohara} \Email{ayumis@tohoku.ac.jp}\\
  \addr{Tohoku University, Sendai, Japan}
}

\begin{document}

\maketitle

\begin{abstract}
It is well known that computing a minimum deterministic finite automaton consistent with a given set of positive and negative examples is NP-hard.
Previous work has identified conditions on the input sample under which the problem becomes tractable or remains hard.
In this paper, we study the computational complexity of the case where the input sample is prefix-closed.
This formulation is equivalent to computing a minimum Moore machine consistent with observations along its runs.
We show that the problem is NP-hard to approximate when the sample set consists of all prefixes of binary strings.
Furthermore, we show that the problem remains NP-hard as a decision problem even when the sample set consists of the prefixes of a single binary string.
Our argument also extends to the corresponding problem for Mealy machines.
\end{abstract}

\section{Introduction}
The problem of learning a deterministic finite automaton (DFA) from a labeled sample set containing positive and negative examples has been a central problem in grammatical inference.
In particular, the \tsc{MinConDFA} problem, which asks for a smallest DFA consistent with the given sample, has been extensively studied.
It was shown to be NP-hard by \cite{Gold2} and \cite{Angluin78}.\footnote{Precisely speaking, the model considered in these works corresponds to a variant of Mealy machines that report the label on the last edge traversed by each input string.  \cite{Lingg} gave a DFA-adaptation of \citeauthor{Gold2}'s proof.}

Furthermore, \cite{PittW1993} showed that \tsc{MinConDFA} is not only NP-hard, but also hard to approximate within any polynomial ratio.
Assuming $\mrm{P} \neq \mrm{NP}$, they proved that for any constant $c$, no polynomial-time algorithm can guarantee a consistent DFA of size at most $m_*^{c}$, where $m_*$ denotes the size of a smallest consistent DFA.
More recently, \cite{ChalermsookLN2014} proved that \tsc{MinConDFA} is NP-hard to approximate within a factor of $n^{1-\epsilon}$ for any constant $\epsilon > 0$, where $n$ is the sample size; this bound is tight, since a trivial $O(n)$-approximation is achieved by the prefix-tree acceptor.

On the other hand, the literature has proposed several polynomial-time algorithms \citep[etc.]{RPNI,EDSM} that compute a smallest DFA when the input sample satisfies a property that is \emph{characteristic} of the target DFA \citep{LgIdent,delaHiguera97}.
This contrasts with \cite{TrakhtenbrotB73}'s algorithm which returns a smallest DFA provided that the input sample is \emph{uniformly complete}, which is a ``structural property'', independently of the target DFA to be learned.
Along this line, \cite{Vazquez2016} showed that \citeauthor{TrakhtenbrotB73}'s algorithm outputs a minimal DFA when the sample is \emph{reasonably complete}, which is a weaker property than the uniform completeness.
Furthermore, \cite{Zhang} introduced an even weaker condition, called \emph{semantic completeness}, and showed that if the input sample is prefix-closed and semantically complete, then the \tsc{MinConDFA} problem can be solved in polynomial time.
These positive results emphasize the role of prefix-closure as a key structural property.
However, \cite{Zhang} also proved that \tsc{MinConDFA} remains NP-hard for prefix-closed samples that are not semantically complete.
In a different but related direction, \cite{ueno2013hardness} studied the case where the input sample consists of some prefixes of a single string and showed that the problem remains NP-complete in this case.

In this paper, we combine the restrictions studied by \cite{Zhang} and \cite{ueno2013hardness}  and consider samples that consist of all prefixes of a single string over a binary alphabet.
We show that even under this unified restriction, the problem remains NP-complete.
For Moore machines, this setting corresponds to observing the input/output labels along a single run, where both the input and the output alphabets are binary.
This models scenarios in which one has trace information from a single execution -- e.g., in black-box testing, runtime monitoring, or trace-based reverse engineering -- but has no ability to reset or systematically explore the state space.
The same NP-completeness result also holds for Mealy machines, where outputs are placed on transitions instead of states. 
Our results show that the \tsc{MinConDFA} problem stays computationally intractable despite the highly constrained structure of the input data.

We also show that the problem remains hard to approximate within any polynomial ratio when the input consists of a prefix-closed set of multiple runs over a binary alphabet, for finding a smallest DFA, acyclic DFA, Moore machine, or Mealy machine.
This strengthens the hardness-of-approximation results of \cite{Shimozono1998}, which were established for acyclic DFAs under general input samples.
Whether such approximation hardness still holds when the input sample is restricted to the prefix-closed sample of a single run remains an open question for future work.
We note that \cite{ueno2013hardness} obtained a hardness-of-approximation result for DFAs in the setting where the input sample consists of some, but not necessarily all, prefixes of a single run.

\section{Preliminaries}
We denote the interval between two integers $i$ and $j$ by $[i,j]=\{\, k \in \mbb{Z} \mid i \le k \le j\,\}$.

For an alphabet $\Sigma$, the set of strings over $\Sigma$ is denoted by $\Sigma^*$.
We write $\varepsilon$ for the empty string.
A string $s$ is a \emph{prefix} of a string $w$ if $w = st$ is the concatenation of $s$ and some string $t$.
By $s^k$ we denote the $k$-fold repetition of $s$.
The set of prefixes of a string $s$ is denoted by $\pre(s)$, which is extended for sets $S$ of strings by $\pre(S) = \bigcup_{s \in S} \pre(s)$. 
The set $S$ is said to be \emph{prefix-closed} if $S = \pre(S)$.
It is \emph{almost prefix-closed} if $S = \pre(S) - \{\varepsilon\}$.
String concatenation is generalized for sets by $ST=\{\,st \in \Sigma^* \mid s \in S \text{ and } t \in T\,\}$.
We write $S^* = \{\, s_1 \dots s_k \mid k \ge 0 \text{ and } s_i \in S \text{ for all $i \in [1,k]$}\,\}$.

A \emph{deterministic finite automaton (DFA)} is a tuple $M=(Q,\Sigma,\delta,F,q_0)$, where
$Q$ is the finite set of states,
$\Sigma$ is the input alphabet, 
$\delta\colon Q \times \Sigma \to Q$ is the state transition function,
$F \subseteq Q$ is the set of accepting states,
and $q_0 \in Q$ is the initial state.
The function $\delta$ is extended to $\delta^* \colon Q \times \Sigma^* \to Q$ by
$\delta^*(q,\varepsilon)=q$ and $\delta^*(q,wa)=\delta(\delta^*(q,w),a)$ for $q \in Q$, $w \in \Sigma^*$, and $a\in\Sigma$.
The DFA induces functions $M_q$ from $\Sigma^*$ to $\{+,-\}$ for each $q \in Q$ defined by $M_q(w) = +$ if and only if $\delta^*(q,w) \in F$.
We often write $M$ for $M_{q_0}$ by identifying the DFA and the induced function $M_{q_0}$.
We say that $M$ \emph{accepts} a string $w$ if $M(w)=+$ and otherwise \emph{rejects} $w$.
A \emph{DFA sample} is a pair $(D_+,D_-)$ of disjoint finite sets of strings.
The sample $(D_+,D_-)$ is \emph{(almost) prefix-complete} if $D_+ \cup D_-$ is (almost) prefix-closed.
A DFA $M$ is \emph{consistent} with $(D_+,D_-)$ if and only if $M(s) = +$ for all $s \in D_+$ and $M(s) = -$ for all $s \in D_-$.

An \emph{acyclic DFA (ADFA)} is a variant of a DFA in which the state transition function can be partial and the state transition graph is acyclic.
That is, there are no state $q$ and nonempty string $s \ne \varepsilon$ such that $\delta^*(q,s)=q$.

A \emph{Moore machine} is a tuple $M=(Q,\Sigma,\Delta,\delta,\rho,q_0)$, where
$Q$ is the finite set of states,
$\Sigma$ is the input alphabet, 
$\Delta$ is the output alphabet, 
$\delta\colon Q \times \Sigma \to Q$ is the state transition function,
$\rho\colon Q \to \Delta$ is the (state) output function,
and $q_0 \in Q$ is the initial state.
The function $\delta$ is extended to $\delta^* \colon Q \times \Sigma^* \to Q$ in the same way as those in DFAs.
The machine induces functions $M_q$ from $\Sigma^*$ to $\Delta^*$ for each $q \in Q$ defined by $M_q(\varepsilon) = \varepsilon$ and $M_q(wa) = M_q(w) \cdot \rho(\delta^*(q,wa))$ for $q \in Q$, $w \in \Sigma^*$, and $a\in\Sigma$.
We often write $M$ for $M_{q_0}$ by identifying the machine with the induced function $M_{q_0}$.
A \emph{run} of $M$ over $s \in \Sigma^*$ is the pair $(s,M(s))$.

A \emph{Mealy machine} $M=(Q,\Sigma,\Delta,\delta,\lambda,q_0)$ is defined similarly to a Moore machine, except for the (transition) output function $\lambda\colon Q \times \Sigma \to \Delta$.
The induced functions $M_q$ for $q \in Q$ from $\Sigma^*$ to $\Delta^*$ are defined by $M_q(\varepsilon) = \varepsilon$ and $M_q(wa) = M_q(w) \cdot \lambda(\delta^*(q,w),a)$.
We often identify $M$ with $M_{q_0}$.
A \emph{run} of $M$ over $s \in \Sigma^*$ is the pair $(s,M(s))$.

A \emph{machine sample set} $D$ is a finite set of string pairs $(s,t) \in \Sigma^* \times \Delta^*$ such that $|s|=|t|$.
The set $D$ is \emph{consistent} if for any $(s_1,t_1),(s_2,t_2) \in D$ and $s \in \pre(s_1) \cap \pre(s_2)$, the prefixes of length $|s|$ of $t_1$ and $t_2$ coincide.

A Moore/Mealy machine $M$ is said to be \emph{consistent} with $D$ if $M(s)=t$ for all $(s,t) \in D$.

\begin{problem}[ConDFA, ConADFA, ConMoore, ConMealy]
  The problem \Con{DFA} is to decide whether there exists a DFA with at most $m$ states consistent with $D$ for a natural number $m$ and a DFA sample $D$.

  The problems \Con{ADFA}, \Con{Moore}, and \Con{Mealy} are similarly defined.
\end{problem}
Those problems are obviously in $\mrm{NP}$.

In this paper, we are concerned with Moore/Mealy machines whose output alphabet $\Delta$ is binary, i.e., $\Delta=\{+,-\}$. 
The problem \Con{DFA} with prefix-complete instances is almost identical to \Con{Moore} with a binary output alphabet, except that the output symbol assigned to the initial state is not included in instances of \Con{Moore}.
Any consistent machine sample $D$ can be converted in polynomial time into an equivalent DFA sample $(D_+,D_-)$, which is almost prefix-complete, and vice versa.

Next, we define the graph coloring problem, one of Karp's original 21 NP-complete problems \citep{Karp1972}. We will reduce from this problem to the DFA learning one in order to prove the latter's NP-hardness.

A \emph{$K$-labeling} of a simple undirected graph $G = (V, E)$ is a map $\Phi\colon V \to [1,K]$.
It is called a \emph{$K$-coloring} if $\{u, v\} \in E$ implies $\Phi(u) \neq \Phi(v)$.
The least number $K$ such that $G$ admits a $K$-coloring is called the \emph{chromatic number} of $G$.
\begin{problem}[Graph Coloring]
    Given a simple undirected graph $G = (V, E)$ and an integer $K$, \tsc{GraphColoring} is to decide whether it admits a $K$-coloring.
\end{problem}

\section{Zhang's reduction}

\cite{Zhang} has shown that \Con{DFA} is NP-complete when the input sample set over an unbounded alphabet is prefix-closed.
We briefly review his proof with a slight modification.

Suppose we are given a graph $G=(V,E)$ and an integer $K$ as an instance of \tsc{GraphColoring}.
For $V=\{v_1,\dots,v_n\}$, we write $e_{ij}=e_{ji} = \{v_i,v_j\}$.
Define a DFA sample $(Z_+,Z_-)$ over $\Sigma = V \cup E$ by
\begin{align*}
Z_+ &= \{\varepsilon\} \cup \{\, v_i e_{ij} \in V E \mid i < j \,\},
\\
Z_- &= V \cup \{\, v_je_{ij} \in V E \mid i < j\,\}.
\end{align*}
The set $Z_+ \cup Z_-$
 is prefix-closed and the total description size is polynomially bounded in the size of the graph.
Then, the input graph is $K$-colorable if and only if there is a DFA with $K+1$ states consistent with $(Z_+,Z_-)$.
\begin{exmpl}\label{ex:zhang}
    Consider the graph shown in Figure~\ref{fig:inputgraph} whose chromatic number is $K=3$.
    The partial DFAs in Figures~\ref{fig:Zhang_trivial} and \ref{fig:Zhang_reduction} are consistent with the sample $(Z_+,Z_-)$ obtained from the graph.
\end{exmpl}
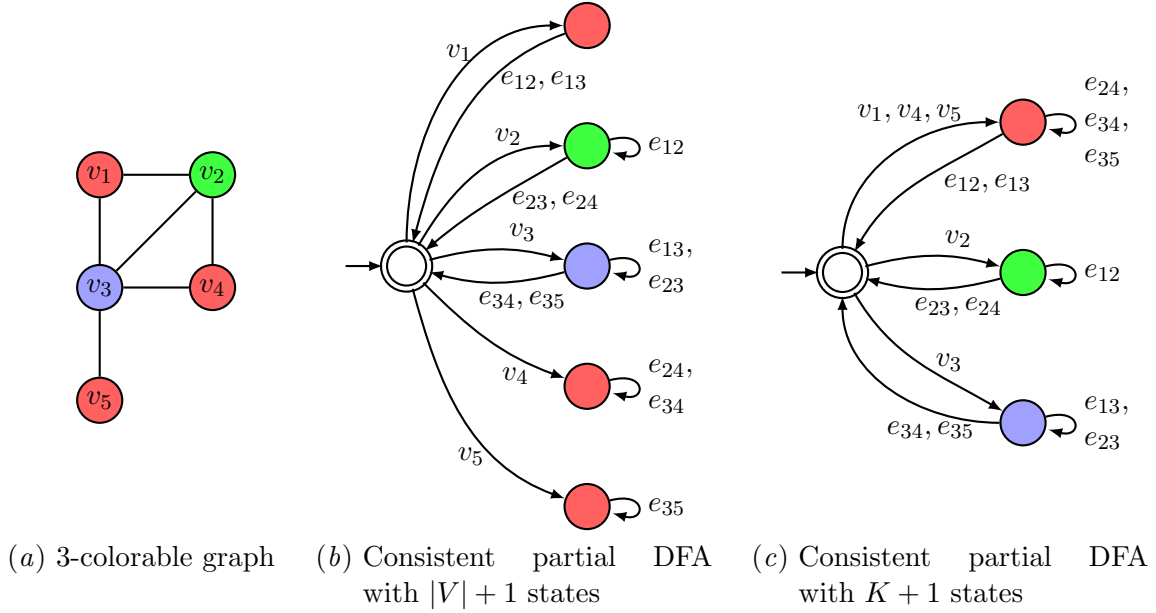
\begin{figure}[tbp]
\floatconts{fig:zhang}
  {\caption{Example of (a modification of) Zhang's reduction.}}
  {
  \subfigure[3-colorable graph\label{fig:inputgraph}]{
    \begin{tikzpicture}[
    vertex/.style={circle, draw=black, thick, minimum size=6mm, inner sep=0pt, font=\bfseries},
    edge/.style={thick}
]
      \node (dumy) at (-1, -1.6) {};
      \node (dummy) at (2.1, -1.6) {};
        \node[vertex, fill=lred] (v1) at (0, 3) {$v_1$};
        \node[vertex, fill=lgreen] (v2) at (1.5, 3) {$v_2$};
        \node[vertex, fill=lblue] (v3) at (0, 1.5) {$v_3$};
        \node[vertex, fill=lred] (v4) at (1.5, 1.5) {$v_4$};
        \node[vertex, fill=lred] (v5) at (0, 0) {$v_5$};
        \draw[edge] (v1) -- (v2);
        \draw[edge] (v1) -- (v3);
        \draw[edge] (v2) -- (v3);
        \draw[edge] (v2) -- (v4);
        \draw[edge] (v3) -- (v4);
        \draw[edge] (v3) -- (v5);
\end{tikzpicture}
  }
 \quad
  \subfigure[Consistent partial DFA with $|V|+1$ states\label{fig:Zhang_trivial}]{
    \begin{tikzpicture}[
    node distance=1.2cm and 1.0cm,
    >=latex,
    thick,
    state_base/.style={circle, draw, minimum size=6mm, inner sep=0pt},
    start_state/.style={state_base, double, double distance=1.5pt, fill=white},
    edge_label/.style={midway, font=\small, text=black}
]
    \node[start_state, initial, initial text=] (start) at (0,0) {};
    \node[state_base,fill=lred] (v1) at (2.4, 3.2) {};
    \node[state_base,fill=lgreen] (v2) at (2.4, 1.6) {};
    \node[state_base,fill=lblue] (v3) at (2.4, 0) {};
    \node[state_base,fill=lred] (v4) at (2.4, -1.6) {};
    \node[state_base,fill=lred] (v5) at (2.4, -3.2) {};
      \draw[->,out=90,in=180] (start) to node[pos=0.7, left] {$v_1$} (v1);
      \draw[<-,out=75,in=200] (start) to node[pos=0.67, below, anchor=west] {$e_{12},e_{13}$} (v1);
      \draw[->,out=60,in=180] (start) to node[pos=0.7, above] {$v_2$} (v2);
      \draw[<-,out=40,in=210] (start) to node[pos=0.56, below=2pt, anchor=west] {$e_{23},e_{24}$} (v2);
      \draw[->,out=15,in=165] (start) to node[pos=0.7, above] {$v_3$} (v3);
      \draw[<-,out=-15,in=195] (start) to node[pos=0.7, below, anchor = north] {$e_{34},e_{35}$} (v3);
      \draw[->,out=-45,in=160] (start) to node[pos=0.7, below] {$v_4$} (v4);
      \draw[->,out=-75,in=160] (start) to node[pos=0.7, left] {$v_5$} (v5);
      \path (v2) edge[loop right] node {$e_{12}$} ();
      \path (v3) edge[loop right] node[align=left] {$e_{13}$,\\$e_{23}$} ();
      \path (v4) edge[loop right] node[align=left] {$e_{24}$,\\$e_{34}$} ();
      \path (v5) edge[loop right] node {$e_{35}$} ();
\end{tikzpicture}
  }
 \quad
  \subfigure[Consistent partial DFA with $K+1$ states\label{fig:Zhang_reduction}]{
    \begin{tikzpicture}[
    node distance=1.2cm and 1.0cm,
    >=latex,
    thick,
    state_base/.style={circle, draw, minimum size=6mm, inner sep=0pt},
    start_state/.style={state_base, double, double distance=1.5pt, fill=white},
    edge_label/.style={midway, font=\small, text=black}
]
    \node (dummy) at (0, -3.3) {};
    \node[start_state, initial, initial text=] (start) at (0,0) {};
    \node[state_base,fill=lred] (r) at (2.4, 2) {};
    \node[state_base,fill=lgreen] (g) at (2.4, 0) {};
    \node[state_base,fill=lblue] (b) at (2.4, -2) {};
      \draw[->,out=90,in=180] (start) to node[pos=0.6, above=4pt] {$v_1,v_4,v_5$} (r);
      \draw[<-,out=60,in=210] (start) to node[pos=0.6, below=4pt, anchor=west] {$e_{12},e_{13}$} (r);
      \draw[->,out=15,in=165] (start) to node[pos=0.7, above] {$v_2$} (g);
      \draw[<-,out=-15,in=195] (start) to node[pos=0.7, below, anchor = north] {$e_{23},e_{24}$} (g);
      \draw[->,out=-60,in=150] (start) to node[pos=0.7, above] {$v_3$} (b);
      \draw[<-,out=-90,in=180] (start) to node[pos=0.7, below, anchor = north] {$e_{34},e_{35}$} (b);
      \path (r) edge[loop right] node[align=left] {$e_{24}$,\\$e_{34}$,\\$e_{35}$} ();
      \path (g) edge[loop right] node[align=left] {$e_{12}$} ();
      \path (b) edge[loop right] node[align=left] {$e_{13}$,\\$e_{23}$} ();
\end{tikzpicture}
  }
  }
\end{figure}

Suppose the input graph $G$ admits a $K$-coloring function $\Phi \colon V \to \{1,\dots,K\}$.
Then, we define a partial DFA $(Q,\Sigma,\delta,F,0)$ by $Q = \{0,1,\dots,K\}$, $F=\{0\}$, and 
\begin{align*}
  \delta(0,v_i) &= \Phi(v_i),
\\
\delta(k,e_{ij}) &= \begin{cases}
    0 & \text{ if $\Phi(i)=k$ and $i<j$,}
\\  k & \text{ if $\Phi(j)=k$ and $i<j$,}
\end{cases}
\end{align*}
for all $k \in \{1,\dots,K\}$, $v_i \in V$, and $e_{ij} \in E$.
The transition function $\delta$ is well-defined since $\Phi(v_i) \ne \Phi(v_j)$ for all $e_{ij} \in E$.
Clearly, the DFA is consistent with $(Z_+,Z_-)$.

Conversely, suppose that $(Z_+,Z_-)$ admits a consistent DFA $M=(Q,\Sigma,\delta,F,q_0)$ with $K+1$ states.
Since $M$ accepts some strings, there are at most $K$ rejecting states.
We may assume without loss of generality that the rejecting states of $M$ are named $1,\dots,K'$ for some $K' \le K$.
Since every $v_i$ is rejected, $\delta(q_0,v_i)$ is a rejecting state.
Define a function $\Phi \colon V \to [1,K']$ by $\Phi(v_i)=\delta(q_0,v_i)$.
For each edge $e_{ij}$ of $G$ with $i < j$, since $v_i e_{ij} \in Z_+$ and $v_j e_{ij} \in Z_-$, we must have $\Phi(v_i)=\delta(q_0,v_i) \ne \delta(q_0,v_j)=\Phi(v_j)$.
Thus, $\Phi$ is a proper $K$-coloring.

For more details, see \citeauthor{Zhang}'s original paper.
\section{Learning DFAs from prefix-complete samples of binary strings}
\label{sec:binary}
This section shows that \Con{DFA} and \Con{ADFA} are NP-hard even to approximate when the input alphabet is binary and the sample is prefix-complete, based on the reduction presented in the previous section.
Our argument implies that \Con{Moore} and \Con{Mealy} over binary alphabets are also hard.

\subsection{NP-hardness of the decision problem}\label{sec:NPcomp_multi}
For an input graph $G=(V,E)$ of \tsc{GraphColoring}, let $\tilde{v}_{i},\tilde{e}_{ij} \in \{0,1\}^*$ be binary encodings\footnote{Collision between $\tilde{v}_i$ and $\tilde{e}_{jk}$ is allowed.} of lengths $\lceil \log_2 |V| \rceil$ and $\lceil \log_2 |E| \rceil$ for $v_i \in V$ and $e_{ij} \in E$, respectively.

Let $L$ be any integer polynomially bounded in $|V|$ and greater than $4|V|+2|E|\lrceil{\log_2|E|}$, and define
\begin{align*}
S &= \{\, \tilde{v}_i 0^L \tilde{e}_{ij} \mid v_i \in V,\ {e}_{ij} \in E \,\},
\\
S_+ &= \{\, \tilde{v}_i 0^L \mid v_i \in V \,\} \cup \{\, \tilde{v}_i 0^L \tilde{e}_{ij} \mid v_i \in V,\ {e}_{ij} \in E,\ i < j \,\},
\\
S_- &= \pre(S) - S_+.
\end{align*}
The sample $(S_+ , S_-)$ is prefix-complete and the total description size is polynomially bounded in the size of the graph.
We note that the above construction is independent of the coloring number $K$.
We call the three substrings $\tilde{v}_i$, $0^L$, and $\tilde{e}_{ij}$ of $\tilde{v}_i 0^L \tilde{e}_{ij} \in S$ the \emph{head}, the \emph{body}, and the \emph{tail}, respectively.

\begin{figure}[tbp]
    \centering
\begin{tikzpicture}[
    node distance=1.2cm and 1.0cm,
    >=latex,
    thick,
    state_base/.style={circle, draw, minimum size=6mm, inner sep=0pt},
    start_state/.style={state_base},
    accept_state/.style={state_base, fill=white, double, double distance=1.5pt},
    white_state/.style={state_base, fill=white},
    red_state/.style={state_base, fill=lred},
    green_state/.style={state_base, fill=lgreen},
    blue_state/.style={state_base, fill=lblue},
    edge0/.style={->, draw=rb},
    edge1/.style={->, draw=gb},
    state1/.style={red_state},
    state2/.style={green_state},
    state3/.style={blue_state},
    state4/.style={red_state},
    state5/.style={red_state},   
    edge_label/.style={midway, font=\small, text=black}
]
    
    \node (head) at (1,-1.6) {Head};
    \draw[dashed] (3.2,-1.8) -- (3.2,8);
    \node[start_state, initial, initial text=] (start) at (0,3.0) {};
    \node[white_state] (u0) at (0.9,4) {};
    \node[white_state] (u00) at (1.8,5.2) {};
    \node[white_state] (u01) at (1.8,2.8) {};
    \node[white_state] (u1) at (0.9,1.8) {};
    \node[white_state] (u10) at (1.8,0.8) {};
    \node[red_state,   label=above:{$\tilde{v}_1$}] (v1) at (3,6.4) {};
    \node[green_state, label=above:{$\tilde{v}_2$}] (v2) at (3,4.8) {};
    \node[blue_state,  label=above:{$\tilde{v}_3$}] (v3) at (3,3.2) {};
    \node[red_state,   label=above:{$\tilde{v}_4$}] (v4) at (3,1.6) {};
    \node[red_state,   label=above:{$\tilde{v}_5$}] (v5) at (3,0) {};

    \draw[edge0] (start) -- node[edge_label, above left=-1mm] {\textcolor{rb}{0}} (u0);
    \draw[edge1] (start) -- node[edge_label, below left=-1mm] {\textcolor{gb}{1}} (u1);
    \draw[edge0] (u0) -- (u00);
    \draw[edge1] (u0) -- (u01);
    \draw[edge0] (u1) -- (u10);
    \draw[edge0] (u00) -- (v1);
    \draw[edge1] (u00) -- (v2);
    \draw[edge0] (u01) -- (v3);
    \draw[edge1] (u01) -- (v4);
    \draw[edge0] (u10) -- (v5);
  
    \node (body) at (6,-1.6) {Body};
    \draw[dashed] (9.2,-1.8) -- (9.2,8);
  \draw[decorate,decoration={brace,amplitude=6pt}] (3.2,7.3) -- (9,7.3) node[midway,above=5pt] {$L$};
    \foreach \i in {1,...,5} {
        \node[state\i, right=0.6cm of v\i] (b\i_1) {};
        \node[state\i, right=0.6cm of b\i_1] (b\i_2) {};
        \node[right=0.4cm of b\i_2] (dots\i) {\textcolor{rb}{$\cdots$}};
        \node[state\i, right=0.4cm of dots\i] (b\i_last) {};
        \node[accept_state,state\i, right=0.6cm of b\i_last] (f\i) {};
        \draw[edge0] (v\i) -- (b\i_1);
        \draw[edge0] (b\i_1) -- (b\i_2);
        \draw[edge0] (b\i_2) -- (dots\i);
        \draw[edge0] (dots\i) -- (b\i_last);
        \draw[edge0] (b\i_last) -- (f\i);
    }
    
    \node (tail) at (10.6,-1.6) {Tail};
    \begin{scope}[xshift=9cm,yshift=-0.2cm]
    \node[accept_state, label=right:{$\tilde{e}_{12}$}] (e12) at (3,8) {};
    \node[accept_state, label=right:{$\tilde{e}_{13}$}] (e13) at (3,7.2) {};
    \node[white_state,  label=right:{$\tilde{e}_{12}$}] (e21) at (3,6.4) {};
    \node[accept_state, label=right:{$\tilde{e}_{23}$}] (e23) at (3,5.6) {};
    \node[accept_state, label=right:{$\tilde{e}_{24}$}] (e24) at (3,4.8) {};
    \node[white_state,  label=right:{$\tilde{e}_{13}$}] (e31) at (3,4) {};
    \node[white_state,  label=right:{$\tilde{e}_{23}$}] (e32) at (3,3.2) {};
    \node[accept_state, label=right:{$\tilde{e}_{34}$}] (e34) at (3,2.4) {};
    \node[accept_state, label=right:{$\tilde{e}_{35}$}] (e35) at (3,1.6) {};
    \node[white_state,  label=right:{$\tilde{e}_{24}$}] (e42) at (3,0.8) {};
    \node[white_state,  label=right:{$\tilde{e}_{34}$}] (e43) at (3,0) {};
    \node[white_state,  label=right:{$\tilde{e}_{35}$}] (e53) at (3,-0.8) {};

    \node[white_state] (f1_0) at (1,7.1) {};
    \node[white_state] (f1_00) at (2,7.4) {};
    \node[white_state] (f2_0) at (1,5.6) {};
    \node[white_state] (f2_00) at (2,6.4) {};
    \node[white_state] (f2_01) at (2,5.6) {};
    \node[white_state] (f3_0) at (1,3.6) {};
    \node[white_state] (f3_00) at (2,4) {};
    \node[white_state] (f3_01) at (2,3.2) {};
    \node[white_state] (f3_1) at (1,2.8) {};
    \node[white_state] (f3_10) at (2,2.4) {};
    \node[white_state] (f4_0) at (1,1.2) {};
    \node[white_state] (f4_01) at (2,1.0) {};
    \node[white_state] (f4_1) at (1,0.4) {};
    \node[white_state] (f4_10) at (2,0.2) {};
    \node[white_state] (f5_1) at (1,-0.6) {};
    \node[white_state] (f5_10) at (2,-0.7) {};
    \end{scope}
    \draw[edge0] (f1) -- (f1_0);
    \draw[edge0] (f1_0) -- (f1_00);
    \draw[edge0] (f1_00) -- (e12);
    \draw[edge1] (f1_00) -- (e13);
    \draw[edge0] (f2) -- (f2_0);
    \draw[edge0] (f2_0) -- (f2_00);
    \draw[edge0] (f2_00) -- (e21);
    \draw[edge1] (f2_0) -- (f2_01);
    \draw[edge0] (f2_01) -- (e23);
    \draw[edge1] (f2_01) -- (e24);
    \draw[edge0] (f3) -- (f3_0);
    \draw[edge1] (f3) -- (f3_1);
    \draw[edge0] (f3_0) -- (f3_00);
    \draw[edge1] (f3_00) -- (e31);
    \draw[edge1] (f3_0) -- (f3_01);
    \draw[edge0] (f3_01) -- (e32);
    \draw[edge0] (f3_1) -- (f3_10);
    \draw[edge0] (f3_10) -- (e34);
    \draw[edge1] (f3_10) -- (e35);
    \draw[edge0] (f4) -- (f4_0);
    \draw[edge1] (f4_0) -- (f4_01);
    \draw[edge1] (f4) -- (f4_1);
    \draw[edge0] (f4_1) -- (f4_10);
    \draw[edge1] (f4_01) -- (e42);
    \draw[edge0] (f4_10) -- (e43);
    \draw[edge1] (f5) -- (f5_1);
    \draw[edge0] (f5_1) -- (f5_10);
    \draw[edge1] (f5_10) -- (e53);
\end{tikzpicture}
    \caption{Tree representation of the sample $(S_+,S_-)$ obtained from the graph in Figure~\ref{fig:inputgraph}.
    The state marked with $\tilde{v}_i$ is reached by reading $\tilde{v}_i$.
    The two states marked with $\tilde{e}_{ij}$ are reached by reading $\tilde{v}_i 0^L \tilde{e}_{ij}$ and $\tilde{v}_j 0^L \tilde{e}_{ij}$.
    The colors of states correspond to the coloring on the graph in Figure~\ref{fig:inputgraph}.}
    \label{fig:consistentPTA}
\end{figure}
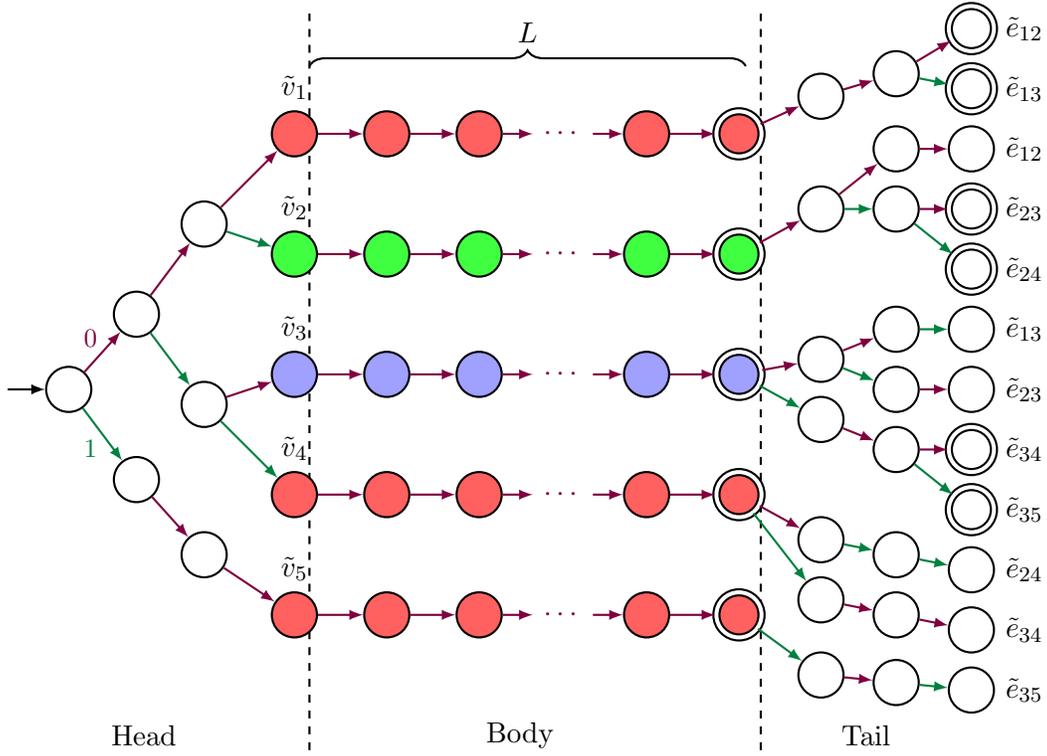
\begin{figure}
    \centering
\begin{tikzpicture}[
    node distance=1.2cm and 1.0cm,
    >=latex,
    thick,
    state_base/.style={circle, draw, minimum size=6mm, inner sep=0pt},
    start_state/.style={state_base},
    accept_state/.style={state_base, fill=white, double, double distance=1.5pt},
    white_state/.style={state_base, fill=white},
    red_state/.style={state_base, fill=lred},
    green_state/.style={state_base, fill=lgreen},
    blue_state/.style={state_base, fill=lblue},
    edge0/.style={->, draw=rb},
    edge1/.style={->, draw=gb},
    state1/.style={red_state},
    state2/.style={green_state},
    state3/.style={blue_state},
    state4/.style={red_state},
    state5/.style={red_state},   
    edge_label/.style={midway, font=\small, text=black}
]
    \newcommand{\pedge}[2]{
      \draw[edge0,transform canvas={yshift=1pt}] (#1) -- (#2);
      \draw[edge1,transform canvas={yshift=-1pt}] (#1) -- (#2);
    }
    \newcommand{\pedgel}[2]{
      \draw[edge0,transform canvas={yshift=1pt}] (#1) -- node[edge_label, above] {\textcolor{rb}{0}} (#2);
      \draw[edge1,transform canvas={yshift=-1pt}] (#1) -- node[edge_label, below] {\textcolor{gb}{1}} (#2);
    }
    
    \node (head) at (1,0) {Head};
    \draw[dashed] (3.2,0) -- (3.2,6);
    \node[start_state, initial, initial text=] (start) at (0,3.2) {};
    \node[white_state] (u0) at (1.1,3.2) {};
    \node[white_state] (u00) at (2.0,4.0) {};
    \node[white_state] (u01) at (2.0,2.4) {};
    \node[green_state] (v2) at (3,4.8) {};
    \node[red_state] (v1) at (3,3.2) {};
    \node[blue_state] (v3) at (3,1.6) {};
    \pedgel{start}{u0}
    \draw[edge0] (u0) -- (u00);
    \draw[edge1] (u0) -- (u01);
    \draw[edge0] (u00) -- (v1);
    \draw[edge1] (u00) -- (v2);
    \draw[edge0] (u01) -- (v3);
    \draw[edge1] (u01) -- (v1);
  
    \node (body) at (6,0) {Body};
    \draw[dashed] (9.2,0) -- (9.2,6);
    \foreach \i in {1,...,3} {
        \node[state\i, right=0.6cm of v\i] (b\i_1) {};
        \node[state\i, right=0.6cm of b\i_1] (b\i_2) {};
        \node[right=0.4cm of b\i_2] (dots\i) {$\cdots$};
        \node[state\i, right=0.4cm of dots\i] (b\i_last) {};
        \node[accept_state,state\i, right=0.6cm of b\i_last] (f\i) {};
        \draw[edge0] (v\i) -- (b\i_1);
        \draw[edge0] (b\i_1) -- (b\i_2);
        \draw[edge0] (b\i_2) -- (dots\i);
        \draw[edge0] (dots\i) -- (b\i_last);
        \draw[edge0] (b\i_last) -- (f\i);
    }
    
    \node (tail) at (10.6,0.0) {Tail};
    \begin{scope}[xshift=9cm]
    \node[accept_state] (a1) at (3,3.7) {};
    \node[white_state] (b1) at (3,2.7) {};
    \node[accept_state] (a2) at (3,4.7) {};
    \node[white_state] (b2) at (3,5.7) {};
    \node[accept_state] (a3) at (3,0.7) {};
    \node[white_state] (b3) at (3,1.7) {};
    
    \node[white_state] (f1_0) at (1,3.7) {};
    \node[white_state] (f1_00) at (2,3.7) {};
    \node[white_state] (f1_1) at (1,2.7) {};
    \node[white_state] (f1_10) at (2,2.7) {};
    \node[white_state] (f2_0) at (1,5.0) {};
    \node[white_state] (f2_00) at (2,5.7) {};
    \node[white_state] (f2_01) at (2,4.7) {};
    \node[white_state] (f3_0) at (1,1.7) {};
    \node[white_state] (f3_00) at (2,1.7) {};
    \node[white_state] (f3_1) at (1,0.7) {};
    \node[white_state] (f3_10) at (2,0.7) {};
    \end{scope}
    \draw[edge0] (f1) -- (f1_0);
    \draw[edge1] (f1) -- (f1_1);
    \draw[edge0] (f1_0) -- (f1_00);
    \draw[edge1] (f1_0) -- (f1_10);
    \draw[edge0] (f1_1) -- (f1_10);
    \pedge{f1_00}{a1};
    \pedge{f1_10}{b1};
    \draw[edge0] (f2) -- (f2_0);
    \draw[edge0] (f2_0) -- (f2_00);
    \draw[edge1] (f2_0) -- (f2_01);
    \draw[edge0] (f2_00) -- (b2);
    \pedge{f2_01}{a2};
    \draw[edge0] (f3) -- (f3_0);
    \draw[edge1] (f3) -- (f3_1);
    \pedge{f3_0}{f3_00};
    \pedge{f3_00}{b3};
    \draw[edge0] (f3_1) -- (f3_10);
    \pedge{f3_10}{a3};
\end{tikzpicture}
    \caption{Consistent ADFA with fewer than $(K+1)L = 4L$ states obtained by merging some states of the tree automaton in Figure~\ref{fig:consistentPTA}.}
    \label{fig:consistentADFA}
\end{figure}
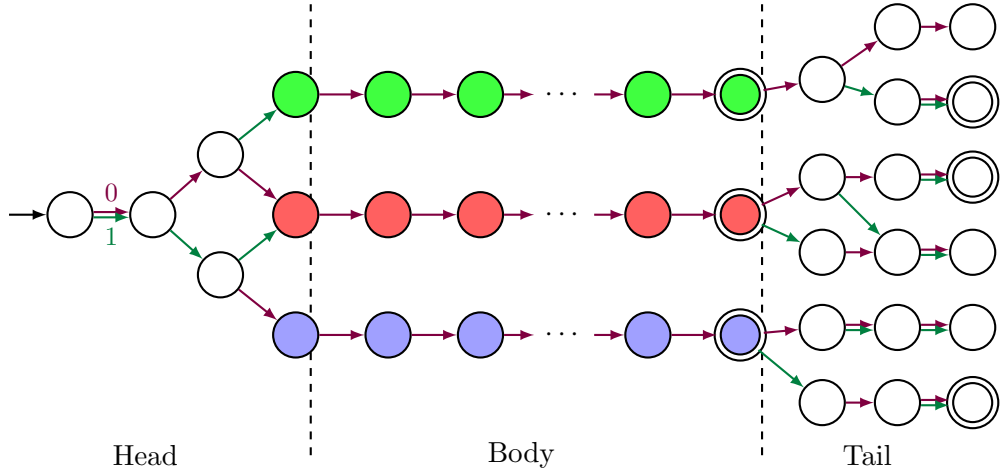
    The tree automaton in Figure~\ref{fig:consistentPTA} represents the sample obtained from the graph in Figure~\ref{fig:inputgraph}, where vertices and edges are encoded as
    $\tilde{v}_1=000$, $\tilde{v}_2=001$, $\tilde{v}_3=010$, $\tilde{v}_4=011$, $\tilde{v}_5=100$,
     $\tilde{e}_{12}=000$, $\tilde{e}_{13}=001$, $\tilde{e}_{23}=010$, $\tilde{e}_{24}=011$, $\tilde{e}_{34}=100$, and $\tilde{e}_{35}=101$.
    The automaton is divided into three parts.
    The strings of $S$ correspond one-to-one with the paths from the initial state to leaf states.
    The group of states reached by reading some prefixes of $\tilde{v}_i$ are called the \emph{head},
    those reached by $\tilde{v}_i 0^h$ for $h \le L$ constitute the \emph{body}, and the rest is the \emph{tail}.
    Let us call a sequence of distinct states $q_1,\dots,q_k$ a \emph{$0$-chain} if $\delta(q_h,0)=q_{h+1}$ for all $h \in [1,k-1]$.
    The body of the tree automaton consists of $|V|$ $0$-chains.
    Each 0-chain is bigger than the total size of the head and the tail.
    Therefore, to substantially reduce the size of the automaton, we must merge the states in the body.
    Recall that $\tilde{v}_i 0^h$ is rejected for every $h < L$ but $\tilde{v}_i 0^L$ is accepted.
    This means that each 0-chain cannot be folded.
    We must unify 0-chains.
    Just like Zhang's reduction, this corresponds to grouping vertices of the graph by distinct color classes.
    By merging the states of the tree automaton in Figure~\ref{fig:consistentPTA}, we obtain the ADFA in Figure~\ref{fig:consistentADFA}.
    This ADFA has fewer than $(K+1)L = 4L$ states.
    Some of its states can be further merged to obtain an even smaller automaton, but the size cannot be reduced to $3L$ or less.
    The number of $0$-chains is the dominant factor.

\begin{lemma}\label{lem:binary}
   A graph $G$ admits a $K$-coloring if and only if $(S_+,S_-)$ admits a consistent DFA with fewer than $m=(K+1)L$ states.
\end{lemma}
\begin{proof}
Suppose $G$ admits a $K$-coloring $\Phi$.
We present a consistent acyclic DFA $M = (Q,\Sigma,\delta,F,q_0)$ with fewer than $(K+1)L$ states.
By arbitrarily defining transition edges missing from $M$, one can obtain a (total) DFA with the same number of states.
The automaton has $K$ $0$-chains of $q_{k,0},\dots,q_{k,L} \in Q$ among which only $q_{k,L}$ is accepting for all $k \in [1,K]$.
We will traverse those $0$-chains when reading the body $0^L$ of each example in $S$.
Concerning the head strings $\tilde{v}_i$, the transition function $\delta$ is defined so that
\[
\delta^*(q_0,\tilde{v}_i) = q_{k,0} \text{ for $k=\Phi(v_i)$}
\,,\]
which requires fewer than $4|V|$ auxiliary states, since this part can be no bigger than the perfect binary tree of height $\lrceil{\log_2 |V|}$.
For the tail strings $\tilde{e}_{ij}$, from the state $q_{k,L}$, the DFA has paths labeled with $\tilde{e}_{ij}$ for $\Phi(v_i) = k$ or $\Phi(v_j)=k$.
The terminal state of the path labeled with $\tilde{e}_{ij}$ from $q_{k,L}$ is accepting if and only if $\Phi(v_i) = k$ and $i < j$.
Since $e_{ij} \in E$ implies $\Phi(v_i) \ne \Phi(v_j)$, such a construction is possible.
At most $2|E| \lrceil{\log_2 |E|}$ states suffice for the tail part since we have $2|E|$ paths each of which has length $\lrceil{\log_2 |E|}$.
All in all, the DFA $M$ has fewer than $4|V|+KL+2|E|\lrceil{\log_2 |E|} \le m$ states.

Conversely, suppose that $(S_+,S_-)$ admits a consistent DFA with at most $m$ states.
Let $q_{i,h} = \delta^*(q_0,\tilde{v}_i 0^h)$ for $i \in [1,n]$ and $h \in [0,L]$.
Since only $\tilde{v}_i 0^L$ is accepted, $q_{i,h} \ne q_{i,h'}$ for any $h < h' \le L$.
On the other hand, it is possible that $q_{i,L} = q_{j,L}$ for distinct $i$ and $j$.
If $q_{i,L} \ne q_{j,L}$, obviously $q_{i,h} \ne q_{j,h'}$ for all $h,h' \le L$.
That is, the 0-chains ending at $q_{i,L}$ and $q_{j,L}$ are disjoint.
Since $M$ has fewer than $(K+1)L$ states, it can contain at most $K$ disjoint 0-chains of length $L$.
This implies that the set $\{\,q_{i,L}\mid i \in [1,n]\,\}$ has at most $K$ elements.
Let us color the vertices of $V$ so that $v_i$ and $v_j$ have the same color if and only if $q_{i,L}=q_{j,L}$.
This gives a $K$-labeling.
If $e_{ij} \in E$ with $i < j$, the fact that $\tilde{v}_i 0^L \tilde{e}_{ij} \in S_+$ and $\tilde{v}_j 0^L \tilde{e}_{ij} \in S_-$ implies $q_{i,L} \ne q_{j,L}$.
That is, $v_i$ and $v_j$ are colored differently.
Therefore, this coloring is a proper $K$-coloring on the graph.
\end{proof}

\begin{theorem}
\label{condfa}
\Con{DFA} and \Con{ADFA} are NP-complete when the input alphabet is binary and the instance sample is prefix-complete.
\Con{Moore} and \Con{Mealy} are NP-complete when the input alphabet is binary.
\end{theorem}
\begin{proof}
It does not matter in the above argument whether the empty string is accepted or rejected.
Thus, the reduction given above proves the NP-hardness of \Con{Moore}.
The same reduction proves the NP-hardness of \Con{Mealy}:
 distinct states are still witnessed by a common suffix yielding different outputs, and a Mealy machine can output $+$ on each transition entering a state that is accepting in the DFA.
\end{proof}

\subsection{Inapproximability}
\cite{PittW1993} showed that polynomial-ratio approximation in polynomial time is hard for the optimization version of \Con{DFA} unless $\mrm{P}=\mrm{NP}$.
We show that it is still the case when we restrict instance samples to be prefix-complete.

\cite{Zuckerman2007} showed that the chromatic number of a graph is not polynomial-time approximable within a polynomial ratio unless $\mrm{P}=\mrm{NP}$.
Suppose that there is a polynomial-approximation algorithm for \Con{DFA} with prefix-complete instances.
That is, there exists $c > 0$ such that for any instance sample, it is guaranteed that 
\[
  \frac{\hat{m}}{m_*} \le m_*^c
\]
where $\hat{m}$ is the size of the output of the algorithm and $m_*$ is the size of the smallest consistent DFA.
Applying the approximation algorithm to the reduction instance in Section~\ref{sec:NPcomp_multi}, Lemma~\ref{lem:binary} implies the graph is $\hat{k}$-colorable for $\hat{k} = \lrfloor{\hat{m}/L}$.

Thus, the chromatic number $k_*$ of the graph is approximated by $\hat{k}$ with ratio
\[
  \frac{\hat{k}}{k_*} = \frac{2\hat{k}L}{2k_*L} \le \frac{ 2\hat{k}L }{(k_*+1)L} < \frac{2\hat{m}}{m_*} \le 2 m_*^c
\]
by $\hat{k}L \le \hat{m}$ and $(k_*+1)L > m_*$.
Since $m_*$ is polynomially bounded in the graph size $n$, the approximation ratio is bounded by a polynomial in $n$.
We obtain the inapproximability theorem for \Con{DFA}.
\begin{theorem}\label{thm:inapproximable}
   Unless $\mrm{P} = \mrm{NP}$, no polynomial-time algorithm approximates \Con{DFA} or \Con{ADFA} within any polynomial ratio on prefix-complete instances over binary alphabets.
   The same holds for \Con{Moore}, and \Con{Mealy} on instances with binary input and output alphabets.
\end{theorem}

\section{Learning DFAs from prefix-complete samples generated by single binary strings}\label{sec:single}

This section shows the NP-hardness of \Con{DFA} when an instance sample set consists of all and only prefixes of a single string over binary alphabets.
Our argument also implies that \Con{Moore} and \Con{Mealy} are NP-complete when the instance is just a single pair of strings over binary alphabets.
Considering \Con{ADFA} under this restriction does not make sense, as the smallest ADFA size must be just the length of the input single string plus one.

Our reduction is a modification of the one we used in \Cref{sec:binary}.
The single string we use here is obtained by concatenating all the strings in $S$ in \Cref{sec:binary} preceded by many consecutive $0$'s.
We set this number of $0$'s to be any integer $N$ such that it is polynomially bounded in $|V|$ and $N > (K+1)L$.
Let
\[
    \Str = \prod_{s \in T} s
    \text{ for }
    T = 0^{N} S
    \,.
\]
The positive examples among $\pre(\Str)$ are those corresponding to $S_+$ and those ending with $s 0^i$ for some $s \in T^*$ and $i \in [1,N-1]$, i.e.,
\[
    D_+ = \pre(\Str) \cap T^* (0^{[1,N-1]} \cup 0^{N} S_+)
\]
where $0^{[i,j]} = \{\, 0^k \mid k \in [i,j]\,\}$.
The negative examples are $D_- = \pre(\Str) - D_+$.
The purpose of the long substring $0^{N}$ is to force the automaton to return to the same state each time an element of $T^*0^{N}$ is read.
When reading a prefix $0^j$ of $0^L$ with $j < L$, we traverse accepting states and then reach a rejecting state when we finish reading $0^L$.
If a consistent DFA with fewer than $2N$ states has two or more $0$-chains of length $N$, most of them are identical -- particularly they end in the same rejecting state.
In contrast to the ``short'' $0$-chains of length $L$ for the body strings, the long chain consists mostly of accepting states except the last one.
As a result, no short $0$-chains can be embedded within the long 0-chain.
Figure~\ref{fig:consistentDFA} shows a DFA consistent with the sample set obtained from the graph in Figure~\ref{fig:inputgraph}.
\begin{figure}
    \centering
        \begin{tikzpicture}[
    node distance=1.2cm and 1.0cm,
    >=latex,
    thick,
    state_base/.style={circle, draw, minimum size=6mm, inner sep=0pt},
    start_state/.style={state_base},
    accept_state/.style={state_base, fill=white, double, double distance=1.5pt},
    white_state/.style={state_base, fill=white},
    red_state/.style={state_base, fill=lred},
    green_state/.style={state_base, fill=lgreen},
    blue_state/.style={state_base, fill=lblue},
    edge0/.style={->, draw=rb},
    edge1/.style={->, draw=gb},
    state1/.style={red_state},
    state2/.style={green_state},
    state3/.style={blue_state},
    state4/.style={red_state},
    state5/.style={red_state},   
    edge_label/.style={midway, font=\small, text=black}
]
    \newcommand{\pedge}[2]{
      \draw[edge0,transform canvas={yshift=1pt}] (#1) -- (#2);
      \draw[edge1,transform canvas={yshift=-1pt}] (#1) -- (#2);
    }
    \newcommand{\pedgel}[2]{
      \draw[edge0,transform canvas={yshift=1pt}] (#1) -- node[edge_label, above] {\textcolor{rb}{0}} (#2);
      \draw[edge1,transform canvas={yshift=-1pt}] (#1) -- node[edge_label, below] {\textcolor{gb}{1}} (#2);
    }
    
    \node[white_state] (start) at (0,3.2) {};
    \node[white_state] (u0) at (1.1,3.2) {};
    \node[white_state] (u00) at (2.0,4.0) {};
    \node[white_state] (u01) at (2.0,2.4) {};
    \node[green_state] (v2) at (3,4.8) {};
    \node[red_state] (v1) at (3,3.2) {};
    \node[blue_state] (v3) at (3,1.6) {};
    \pedgel{start}{u0}
    \draw[edge0] (u0) -- (u00);
    \draw[edge1] (u0) -- (u01);
    \draw[edge0] (u00) -- (v1);
    \draw[edge1] (u00) -- (v2);
    \draw[edge0] (u01) -- (v3);
    \draw[edge1] (u01) -- (v1);
  
    \foreach \i in {1,...,3} {
        \node[state\i, right=0.6cm of v\i] (b\i_1) {};
        \node[state\i, right=0.6cm of b\i_1] (b\i_2) {};
        \node[right=0.4cm of b\i_2] (dots\i) {\textcolor{rb}{$\cdots$}};
        \node[state\i, right=0.4cm of dots\i] (b\i_last) {};
        \node[accept_state,state\i, right=0.6cm of b\i_last] (f\i) {};
        \draw[edge0] (v\i) -- (b\i_1);
        \draw[edge0] (b\i_1) -- (b\i_2);
        \draw[edge0] (b\i_2) -- (dots\i);
        \draw[edge0] (dots\i) -- (b\i_last);
        \draw[edge0] (b\i_last) -- (f\i);
    }
    \draw[decorate,decoration={brace,amplitude=6pt}] (3.2,5.2) -- (9.2,5.2) node[midway,above=5pt] {$L$};
    
    \begin{scope}[xshift=9cm]
    \node[accept_state] (a1) at (3,3.7) {};
    \node[white_state] (b1) at (3,2.7) {};
    \node[accept_state] (a2) at (3,4.7) {};
    \node[white_state] (b2) at (3,5.7) {};
    \node[accept_state] (a3) at (3,0.7) {};
    \node[white_state] (b3) at (3,1.7) {};

    \node[white_state] (f1_0) at (1,3.7) {};
    \node[white_state] (f1_00) at (2,3.7) {};
    \node[white_state] (f1_1) at (1,2.7) {};
    \node[white_state] (f1_10) at (2,2.7) {};
    \node[white_state] (f2_0) at (1,5.0) {};
    \node[white_state] (f2_00) at (2,5.7) {};
    \node[white_state] (f2_01) at (2,4.7) {};
    \node[white_state] (f3_0) at (1,1.7) {};
    \node[white_state] (f3_00) at (2,1.7) {};
    \node[white_state] (f3_1) at (1,0.7) {};
    \node[white_state] (f3_10) at (2,0.7) {};
    \end{scope}
    \draw[edge0] (f1) -- (f1_0);
    \draw[edge1] (f1) -- (f1_1);
    \draw[edge0] (f1_0) -- (f1_00);
    \draw[edge1] (f1_0) -- (f1_10);
    \draw[edge0] (f1_1) -- (f1_10);
    \pedge{f1_00}{a1};
    \pedge{f1_10}{b1};
    \draw[edge0] (f2) -- (f2_0);
    \draw[edge0] (f2_0) -- (f2_00);
    \draw[edge1] (f2_0) -- (f2_01);
    \draw[edge0] (f2_00) -- (b2);
    \pedge{f2_01}{a2};
    \draw[edge0] (f3) -- (f3_0);
    \draw[edge1] (f3) -- (f3_1);
    \pedge{f3_0}{f3_00};
    \pedge{f3_00}{b3};
    \draw[edge0] (f3_1) -- (f3_10);
    \pedge{f3_10}{a3};
    
    \node[start_state, initial, initial text=, initial where=right] (r0) at (13,7) {};
    \node[accept_state] (r1) at (12,7) {};
    \node[accept_state] (r2) at (11,7) {};
    \node[accept_state] (r3) at (10,7) {};
    \node[] (r4) at (9,7) {};
    \node[] (r5) at (3,7) {};
    \node[accept_state] (r6) at (2,7) {};
    \node[accept_state] (r7) at (1,7) {};
    \node[accept_state] (r8) at (0,7) {};
    \foreach \i in {1,...,3} {
        \draw[edge0,-] (a\i) -- +(1cm,0);
        \draw[edge0,-] (b\i) -- +(1cm,0);
    }
    \draw[edge0] (a3) -- +(1cm,0) -- node[edge_label, right] {\textcolor{rb}{0}} +(1cm,5.6cm) -- (r1);
    \draw[edge0] (r0) -- node[edge_label, above] {\textcolor{rb}{0}} (r1);
    \draw[edge0] (r1) -- (r2);
    \draw[edge0] (r2) -- (r3);
    \draw[edge0] (r3) -- (r4);
    \draw[edge0,-,dotted] (r4) -- (r5);
    \draw[edge0] (r5) -- (r6);
    \draw[edge0] (r6) -- (r7);
    \draw[edge0] (r7) -- (r8);
    \draw[edge0] (r8) -- (start);
    \node (r3) at (10,7) {};
    \draw[decorate,decoration={brace,amplitude=6pt}]
      ($(r8)+(-3mm,4mm)$) -- ($(r0)+(-3mm,4mm)$)
      node[midway,above=5pt] {$N-1$};
  \end{tikzpicture}
    \caption{DFA consistent with the sample set obtained from the graph in Figure~\ref{fig:inputgraph}.}
    \label{fig:consistentDFA}
\end{figure}
\begin{lemma}\label{lem:single}
   A graph $G$ admits a $K$-coloring if and only if $(D_+,D_-)$ admits a consistent DFA with at most $m=N+(K+1)L$ states.
\end{lemma}
\begin{proof}
Suppose that the graph admits a proper $K$-coloring.
We give a consistent DFA $M'$ by modifying the DFA $M$ given in the proof of Lemma~\ref{lem:binary}.
We introduce $N$ new states $q_{0,0},\dots,q_{0,N-1}$ that form a $0$-chain together with the initial state of $M$, where $q_{0,0}$ is the new initial state and the $0$-chain ends in the initial state of $M$.
In addition, we add a $0$-transition from every leaf state $q$ of $M$ to $q_{0,1}$.
Here, the leaf states are those reached by reading elements of $S$ in $M$. 
This is a consistent DFA with $N$ more states than $M$.

Conversely, suppose there is a consistent DFA $M$ with at most $N+(K+1)L$ states.
The presence of the prefix $0^{N}$ where $0^{N}$ is rejected and $0^j$ are accepted for all $j \in [1,N-1]$ implies that $M$ must have a $0$-chain of $q_{0,1},\dots,q_{0,N}$ where $\delta^*(q_0,0^i) = q_{0,i}$ for the initial state $q_0$ for all $i \in [1,N]$, $q_{0,i}$ are accepting for all $i \in [1,N-1]$, and $q_{0,N}$ is rejecting.
Consider the state reached by reading a string of the form $s 0^{N} \in \pre(\Str)$ for some $s \in T^*$.
If the state is not $q_{0,N}$, we must have a 0-chain of length $N$ disjoint from $q_{0,1},\dots,q_{0,N}$, which means that the DFA has more than $2 N > m$ states, a contradiction.
Therefore, each time we read a string in $T^* 0^{N}$, we must come back to the same state.
Hence, the same argument as the proof of Lemma~\ref{lem:binary} applies, and we conclude that the graph $G$ is $K$-colorable.
Note that the rejecting states $q_{k,0},\dots,q_{k,L-1} \notin F$ of the short $0$-chains for $k \in [1,K]$ and the accepting states $q_{0,1},\dots,q_{0,N-1} \in F$ of the long $0$-chain are disjoint, since they carry different accept/reject labels.
Therefore, no short $0$-chains can be embedded into the long $0$-chain to evade the counting argument.
\end{proof}

\begin{theorem}
\label{thm:single}
\Con{DFA} and \Con{ADFA} are NP-complete when the input alphabet is binary and the instance sample consists of all and only the prefixes of a single string.
\Con{Moore} and \Con{Mealy} are NP-complete when the input is a single run over a binary alphabet.
\end{theorem}

We note that, unlike Lemma~\ref{lem:binary}, Lemma~\ref{lem:single} cannot be extended to show polynomial-ratio inapproximability.
 In our reduction, any consistent DFA requires at least $N$ states, while the trivial $|V|$-coloring gives a consistent DFA with at most $N+(|V|+1)L$ states.
 Since the reduction for the approximation problem must set $N > (|V|+1)L$ (as no target value $K$ is given in the approximation problem), this trivial consistent DFA has fewer than twice the number of states of a smallest consistent DFA.
Note that even if we could use a tighter known upper bound $\hat{k} \le |V|$ on the chromatic number and set $N = (\hat{k}+1)L + 1$, there would exist a consistent DFA with at most $N+(\hat{k}+1)L$ states, which is again less than twice the size of a smallest consistent DFA.
This approach does not even rule out a factor-2 approximation.
\section{Conclusion and future work}
Expanding upon previous work, we have restricted the types of instances handled in problems that aim to find machines of minimum number of states consistent with those instances. Using a previous reduction from a graph problem, we proved that learning (acyclic) DFAs, as well as Moore and Mealy machines, from prefix-complete instances in a binary alphabet, when the instance strings are all prefixes of a single string, is NP-complete. If the instance is binary and prefix-closed, we have also proved, assuming $\mrm{P} \neq \mrm{NP}$, that the above-mentioned machine learning problems are not polynomially approximable with any polynomial-time algorithm.

The main question left open by our work is whether the inapproximability of Theorem~\ref{thm:inapproximable} survives the single-string restriction of Section~\ref{sec:single}.
Here the two requirements appear to be in tension, at least under our approach:
the long $0^N$ padding that makes the single-string reduction possible necessarily dominates the automaton's size, swamping the coloring-dependent term and capping the achievable approximation ratio at two.
Establishing inapproximability would require a construction in which the coloring-dependent part, rather than a fixed dominating chain, governs the optimum size;
 whether such a construction exists, or whether the single-string case instead admits a good approximation, remains unclear.
 The fact that the prefixes of a single string form a path rather than a tree suggests recasting the problem as a constrained partition of $[0,|\Str|]$, which may be more amenable to approximation.

 More broadly, since our results push the natural restrictions on the input quite far, the complementary task of designing algorithms that solve or approximate the problem efficiently on well-behaved samples seems the more promising long-term direction.

\bibliography{refs}

\end{document}